\documentclass[11pt]{article}

\usepackage[margin=1in]{geometry}
\usepackage{times}

\usepackage{graphicx}
\usepackage{theorem,latexsym,amssymb}
\usepackage[small]{caption}
\usepackage{graphicx}
\usepackage{amsmath}
\usepackage{booktabs}
\usepackage{natbib}
\usepackage{color}

\newcommand{\tabincell}[2]{\begin{tabular}{@{}#1@{}}#2\end{tabular}}

\usepackage[ruled,linesnumbered,vlined]{algorithm2e}

\newtheorem{claim}{Claim}[section]
\newtheorem{theorem}{Theorem}[section]

\newtheorem{lemma}[theorem]{Lemma}

\newtheorem{example}[theorem]{Example}

\newenvironment{proof}{{\bf Proof:  }}{\hfill\rule{2mm}{2mm}\vspace*{5pt}}
\newenvironment{proofof}[1]{{\vspace*{5pt} \noindent\em Proof of #1:  }}{\hfill\rule{2mm}{2mm}\vspace*{5pt}}

\newcommand{\floor}[1]{\lfloor #1 \rfloor }
\newcommand{\MMS}{\mathsf{MMS}}

\newcommand{\RD}{\mathsf{RandomDecline}}
 
\newcommand{\CP}{\mathsf{ConsecutivePick}}
\newcommand{\SRR}{\mathsf{SesquiRR}}

\title{Approximate and Strategyproof Maximin Share Allocation of Chores with Ordinal Preferences\thanks{Part of the results of this article appeared in IJCAI 2019~\citep{ijcai/AzizLW19}. This extended version contains several new results and some analyses that are more refined.} }

\author{
	Haris Aziz$^1$ \qquad
	Bo Li$^2$ \qquad
	Xiaowei Wu$^3$\\
	$^1$ UNSW Sydney and Data61 CSIRO, Australia\\
	$^2$ Department of Computing, The Hong Kong Polytechnic University, Hong Kong, China\\
	$^3$ IOTSC, University of Macau, Macau, China\\
	\texttt{\small haziz@cse.unsw.edu.au, comp-bo.li@polyu.edu.hk, xiaoweiwu@um.edu.mo}
}

\begin{document}

\maketitle

\begin{abstract}
We initiate the work on maximin share (MMS) fair allocation of $m$ indivisible chores to $n$ agents using only their ordinal preferences, from both algorithmic and mechanism design perspectives. 
The previous best-known approximation is $2-1/n$ by Aziz et al. [IJCAI 2017].
We improve this result by giving a simple deterministic $5/3$-approximation algorithm that determines an allocation sequence of agents,  according to which items are allocated one by one.
By a tighter analysis, we show that for $n=2,3$, our algorithm achieves better approximation ratios, and is actually optimal.
We also consider the setting with strategic agents, where agents may misreport their preferences to manipulate the outcome. 
We first provide a $O(\log (m/n))$-approximation consecutive picking algorithm, and then improve the approximation ratio to $O(\sqrt{\log n})$ by a randomized algorithm. 
Our results uncover some interesting contrasts between the approximation ratios achieved for chores versus goods.
\end{abstract}

\section{Introduction} \label{intro}
Multi-agent resource allocation and fair division are major themes in computer science~\citep{choice/BouveretCM16} and mathematical economics~\citep{BrTa96a,BaYo01a}. In this work, we consider allocation algorithms to fairly assign $m$ heterogeneous and indivisible chores to $n$ agents with additive utilities. We take both algorithmic and mechanism design perspectives. 
Firstly, we explore how well we can achieve fairness guarantees when only considering ordinal preferences.  There is a growing body of work on this issue~\citep{DBLP:conf/wine/AnshelevichS16,DBLP:journals/sigecom/Anshelevich16,DBLP:conf/aaai/AnshelevichS16} where it being explored how well ordinal information can help approximate objectives based on cardinal valuations.
Secondly, we take a mechanism design perspective to the problem of fair allocation. We impose the requirement that the algorithm should be strategyproof, i.e., no agent should have an incentive of reporting untruthfully. Under this requirement, we study how well the fairness can be approximated. 
This approach falls under the umbrella of approximation mechanism design without money that has been popularized by \citep{teco/ProcacciaT13}.

The fairness concept we use in this paper is the intensively studied and well-established maximin share fairness. The {\em maximin fair share} (MMS) of an agent is the best she can guarantee if she is allowed to partition items into $n$ bundles but then receives the least preferred one, which was proposed by \citep{bqgt/Budish10} as a fairness concept for allocation of indivisible items. The concept coincides with the standard proportionality fairness concept if the items are divisible. 
It has been proved by \citep{sigecom/ProcacciaW14} and \citep{jacm/KurokawaPW18} that there may not exist an allocation such that very agent's utility is no worse than her MMS. 
As a result, significant effort has been focused on algorithms that find approximate MMS allocations~\citep{icalp/AmanatidisMNS15,jacm/KurokawaPW18}. 
In recent years, \citep{sigecom/GargT20} and \citep{corr/HuangL19} obtained algorithms to find a state of the art ($4/3-\Theta(1/n))$- and 11/9-approximate MMS fair allocations for goods and chores respectively. 
For a more detailed literature view, please refer to Section~\ref{sec:review}.

On one hand, one agent's MMS is defined with respect to her cardinal preference, which places an exact numerical value on each item, and all the aforementioned works assume that the algorithm has full information of these cardinal values.
Since cardinal values can sometimes be difficult to obtain,
this has led researchers to study {\em ordinal algorithms} which only ask agents to rank the goods in the order of their preferences, i.e. the ordinal preferences \citep{ecai/BouveretEL10,ai/AzizGMW15}.
A decision maker wants to know what the price of the missing information is by knowing only ordinal preferences.
\citep{ijcai/AmanatidisBM16} proved that with only ordinal information about the valuations, no algorithm can guarantee better than $\Omega(\log n)$-approximation (for goods). Very recently, \citep{halperndistortion} showed that there is an ordinal  algorithm that guarantees $O(\log n)$-approximate MMS fairness for all agents. 
These works only focused on the case of goods, but there are many settings in which agents may have negative utilities such as when chores or tasks are to be allocated.
In this work, we study to what extent MMS fairness can be guaranteed via ordinal preferences when the items are chores.

In the works discussed above, the focus has been on examining the existence or approximation of MMS allocations. In other words, the problem has been considered from an algorithmic point of view but incentive compatibility has not been addressed. Strategic agents may have incentives to misreport their preferences to manipulate the final allocation of the algorithm in order to increase their utilities. 
Accordingly, a natural question is if it is possible to elicit truthful preferences and also guarantee approximate MMS fairness? Strategyproofness can be a demanding constraint especially when monetary transfers are not allowed.
\citep{ijcai/AmanatidisBM16} were the first to embark on a study of strategyproof and approximately MMS fair algorithms.
They gave a deterministic strategyproof ordinal algorithm which is $O(m-n)$-approximate when the items are goods.
In this paper, we revisit strategyproof MMS allocation by considering the case of chores.
All in all, in this work, we want to answer the following research questions.
\begin{quote}
    {\em When allocating indivisible chores, what approximation guarantee of maximin share fairness can be achieved using ordinal preferences? 
    Furthermore, how can we elicit agents' true preferences and still approximate maximin share fairness?}
\end{quote}

\subsection{Our results}

\paragraph{Algorithmic Perspective.} 
We first take an algorithmic perspective on fair allocation of indivisible chores to agents using ordinal preferences. With cardinal preferences, the best known result is the 11/9-approximate MMS algorithm by \citep{corr/HuangL19}. 
We note that the round-robin algorithm that uses only agents' ordinal preferences returns $2 - 1/n$ approximate MMS allocations \citep{aaai/AzizRSW17}.
In this work, we first improve this result by designing a simple periodic sequential allocation algorithm that ensures $5/3$ approximation for all $n$. 
Interestingly, by refining our analyses and constructing hard instances for $n=2,3$,  we show that our algorithm is actually optimal for these cases.
\begin{table}[htbp]
	\begin{center}
		\begin{tabular}{ |c||cc|cc|} 
			\hline
			& \multicolumn{2}{c|}{Goods}  & \multicolumn{2}{c|}{Chores}  \\ \hline
			& Lower & Upper  & Lower  & Upper \\ \hline \hline
			Ordinal & {\tabincell{c}{$H_n$ \\  \citeauthor{ijcai/AmanatidisBM16}\\{[2016]}}}   & {\tabincell{c}{$2H_n$ \\ \citeauthor{halperndistortion}\\{[2020]}}}  & {\tabincell{c}{$4/3$ for $n = 2$ \\ $7/5$ for $n = 3$ \\  \text{[Our work]}}} & {\tabincell{c}{$4/3$ for $n = 2$ \\ $7/5$ for $n = 3$ \\ $5/3$ for $n \ge 4$ \\  \text{[Our work] }}} \\\hline
			Cardinal & Unknown & {\tabincell{c}{$4/3 - \Theta(1/n)$ \\ \citeauthor{sigecom/GargT20}\\{[2020]}}}  & Unknown & {\tabincell{c}{$11/9$ \\ \citeauthor{corr/HuangL19}\\{[2019]}}}  \\ \hline
		\end{tabular}
	\end{center}
	\caption{
		Lower and upper bounds on approximation of MMS fairness for allocating goods or chores using cardinal or ordinal preferences. Here $H_n = \Theta(\log n)$ is the $n$-th harmonic number and $n$ is the number of agents.}
	\label{table:algorithm}
\end{table}

Our results depend on the following two ideas.
Firstly, we reduce any chore allocation instance to a special one where all agents have the same ordinal preference for items, which is essentially the hardest situation for maximin share fair allocation. The technique has been used previously  \citep{journals/aamas/BouveretL16,journals/teco/BarmanK20,corr/HuangL19}.
Secondly, our algorithm falls under the umbrella of  sequential allocating algorithms in which items are ordered in decreasing order of their costs and assigned to agents sequentially following the order.
In particular, we consider allocation sequences that have a \emph{pattern} and the sequence is obtained by repeating the pattern.
We design a pattern with a length of roughly $1.5n$, and name our algorithm as the {\em Sesqui-Round Robin} Algorithm.
While we prove that our algorithm is optimal for $n \le 3$, we note that it is not optimal for larger $n=4$ (for a detailed discussion, please refer to Section~\ref{sec:conclusion}).
We leave exploring the optimal algorithm for arbitrary $n$ as further study.

\paragraph{Mechanism Design Perspective.} 
We also take a mechanism design perspective for our problem when the agents may misreport their preferences to decrease costs.  
We design a deterministic sequential picking algorithm,  $\CP$, where each agent consecutively selects a number of items, and show that it is strategyproof.
Roughly speaking, given an order of the agents, $\CP$ lets each agent $i$ pick $a_{i}$ items and leave, where $\sum_{i} a_i = m$.
\citep{ijcai/AmanatidisBM16} proved that when the items are goods, the best $\CP$ algorithm can guarantee an approximation of $\floor{(m-n+2)/2}$,
and such an approximation can be easily achieved by letting each of the first $n-1$ agents select one item and allocating all the remaining items to the last agent.
Compared to their result, we show that by carefully deciding the $a_{i}$'s, when items are chores,
we are able to significantly improve the bound to $O(\log (m/n))$\footnote{In this paper we use $\log(\cdot)$ to denote $\log_2(\cdot)$.}.
Moreover, we show that this approximation ratio is the best a $\CP$ algorithm can achieve. 
We further improve the approximation ratio by randomized algorithms.
Particularly, we show that by randomly allocating each item but allowing each agent to
reject a small set of ``bad'' items (i.e., with the largest cost) once, the resulting algorithm is strategyproof and achieves an approximation ratio of $O(\sqrt{\log n})$ in expectation.

\paragraph{Organization.}
We formally define our model and introduce necessary notations in Section~\ref{sec:preli}.
The algorithmic results for approximating MMS fairness using ordinal preferences are given in Section~\ref{sec:algorithmic}.
We present strategyproof algorithms in Section~\ref{sec:sp} and a detailed literature review in Section~\ref{sec:review}.
Finally, Section~\ref{sec:conclusion} concludes the paper with some discussions on the future works.

\section{Model and Preliminaries}\label{sec:preli}

In a fair allocation problem, $N$ is a set of $n$ agents, and $M$ is a set of $m$ indivisible items.
The goal is to fairly distribute all the items to these agents.
Different agents may have different preferences for these items and these preferences are generally captured by
utility or {\em valuation} functions: each agent $i$ is associated with a function $v_i:2^{M}\to \mathbb{R}$ that valuates any set of items.

\paragraph{MMS fairness.}

Imagine that agent $i$ gets the opportunity to partition all items into $n$ bundles, but she is the last to choose a bundle.
Then her best strategy is to partition the items such that the smallest value of a bundle is maximized.
Let $\Pi(M)$ denote the set of all $n$-partitionings of $M$.
Then the {\em maximin share (MMS)} of agent $i$ is defined as
\begin{equation}\label{eq:mms:negative}
\MMS_i = \max_{( X_1, \ldots, X_n) \in \Pi(M)} \min_{j \in N}  v_i(X_j).
\end{equation}
If agent $i$ receives a bundle of items with value at least $\MMS_i$, this allocation is called MMS fair to her.

\medskip

In this work, it is assumed that items are chores: $v_i(S)\leq 0$ for all $i\in N$ and $S \subseteq M$.
Then each agent actually wants to receive as few items as possible.
For ease of description, we ascribe a disutility or \emph{cost} function $c_i=-v_i$ for each agent $i$.
We further assume that the cost function of each agent $i$ is additive.
Accordingly, the cost function $c_{i}$ can be represented by a cost vector $(c_{i1},\ldots,c_{im})$ where $c_{ij}=c_{i}(\{j\})$ is the cost of agent $i$ for item $j$.
Then for any $S\subseteq M$ we have $c_{i}(S)=\sum_{j\in S}c_{ij}$.
We refer $c = (c_1, \ldots, c_n)$ as the {\em cardinal} preference profile. 
Agent $i$'s maximin share can be equivalently defined as
\begin{equation}\label{eq:mms:positive}
\MMS_i = \min_{ (X_1, \ldots, X_n) \in \Pi(M)} \max_{j \in N}  c_i(X_j).
\end{equation}

Note that the maximin threshold defined in Equation \ref{eq:mms:positive} is positive which is
the opposite number of the threshold defined in Equation \ref{eq:mms:negative}.
Throughout the rest of our paper, we choose to use the second definition.
For each agent $i$, we use a permutation over $M$, $\sigma_{i}:[m] \to M$, to denote agent $i$'s {\em ranking} on the items: $c_{i\sigma_{i}(1)}\geq\ldots \geq c_{i\sigma_{i}(m)}$.
In other words, item $\sigma_i(1)$ is the least preferred item and $\sigma_i(m)$ is the most preferred.
We refer to $\sigma = (\sigma_1, \ldots, \sigma_n)$ as the {\em ordinal} preference profile. 
Let $x=(x_{i})_{i\in N}$ be an {\em allocation}, where $x_{i}=(x_{ij})_{j\in M}$ and $x_{ij}\in \{0,1\}$ indicates if agent $i$ gets item $j$ under allocation $x$.
A feasible allocation guarantees a partition of $M$, i.e., $\sum_{i\in N}x_{ij}=1$ for any $j\in M$.
We somewhat abuse the definition and let $X=(X_i)_{i\in N}$, $X_{i}=\{j\in M : x_{ij}=1\}$ and $c_{i}(x)=c_{i}(x_{i})=c_{i}(X_{i})$.
We call an allocation an {\em $\alpha$-MMS allocation} if $c_{i}(x_{i}) \leq \alpha\cdot\MMS_{i}$ for all agents $i$.
When $\alpha = 1$, the allocation is called an {\em MMS allocation}.

We first state the following simple observation about MMS.
Lemma~\ref{lem:mms:bound} implies if an agent receives $k$ items, then her cost is at most $k\cdot \MMS_i$.

\begin{lemma} \label{lem:mms:bound}
	For any agent $i$ and any cost function $c_{i}$, we have
	\begin{itemize}
		\item $\MMS_{i}\geq \frac{1}{n}\cdot c_{i}(M)$;
		\item $\MMS_{i}\geq c_{ij}$ for any $j\in M$.
	\end{itemize}
\end{lemma}
\begin{proof}
	The first inequality is clear as for any partition of the items, the largest bundle has cost at least the average of total cost, i.e., $\frac{1}{n}\cdot c_{i}(M)$.
	For the second inequality, it suffices to show $\MMS_{i}\geq c_{i\sigma_{i}(1)}$.
	This is also clear since in any partition of the items, 
	$\sigma_{i}(1)$ belongs to some bundle and thus the largest bundle should have cost at least $c_{i \sigma_{i}(1)}$.
\end{proof}

By Lemma \ref{lem:mms:bound}, it is easy to see that if $m\leq n$, any allocation that allocates at most one item to each agent is an MMS allocation. Thus throughout this paper, we assume $m > n$.

\paragraph{Ordinal Algorithm.}
An \emph{ordinal algorithm} ${\cal A}$ takes the ordinal preferences $\sigma$ of agents (instead of cardinal preferences $c$) as input, and computes an allocation ${\cal A}(\sigma)$.
Note that the agents do have cardinal cost functions, according to which $\MMS_i$'s are defined.
We call an ordinal algorithm $\alpha$-approximate MMS if for any cost functions $c$ that are consistent with the ordinal preference $\sigma$, the allocation ${\cal A}(\sigma)$ given by the algorithm is an $\alpha$-$\MMS$ allocation.
That is $c_i({\cal A}(\sigma)) \le \alpha \cdot \MMS_i$ for all $i$.
A randomized algorithm ${\cal A}$ returns a distribution over $\Pi(M)$ and is called $\alpha$-approximate MMS if for any cost functions (consistent with the ordinal ranking) $c_{1},\ldots, c_{n}$,
\begin{equation*}
\mathbf{E}_{x\sim \mathcal{A}(\sigma)} \left[\max_{i\in N}\frac{c_{i}(x)}{\MMS_{i}}\right] \leq \alpha.
\end{equation*}

\paragraph{Remark.}
It is necessary and more interesting to define the approximation as the expectation of the maximum ratio over all agents.
If the $\alpha$-approximation is defined as for every agent $i$, $\mathbf{E}_{x\sim \mathcal{A}(\sigma)}c_{i}(x)\leq \alpha\cdot\MMS_{i}$, the problem becomes trivial as uniform-randomly allocating all items gives an exact $\MMS$ allocation.

\paragraph{Strategyproof Algorithm.}
In this work, we also study the situation when the cost rankings $\sigma_i$ are private preferences of agents.
Each agent may misreport her true ranking in order to minimize her own cost for the allocation.
We call an algorithm \emph{strategyproof} if no agent can unilaterally misreport her ranking to reduce her cost.
Formally, a deterministic algorithm $\mathcal{A}$ is called {\em strategyproof} if for every agent $i$, ranking $\sigma_{i}$ and the ranking profile $\sigma_{-i}$ of other agents, 
\[
c_{i}(\mathcal{A}(\sigma_{i},\sigma_{-i}))\leq c_{i}(\mathcal{A}(\sigma'_{i},\sigma_{-i})) \text{ holds for all $\sigma'_{i}$.}
\]
We call a randomized algorithm $\mathcal{A}$ {\em strategyproof in expectation} if for every $i$, $\sigma_{i}$ and $\sigma_{-i}$, 
\[
\mathbf{E}_{x\sim \mathcal{A}(\sigma_{i},\sigma_{-i})}c_{i}(x)
\leq \mathbf{E}_{x\sim \mathcal{A}(\sigma'_{i},\sigma_{-i})} c_{i}(x) \text{ holds for all $\sigma'_{i}$.}
\]

\section{Approximate Maximin Share with Ordinal Preferences}
\label{sec:algorithmic}

In this section we consider the problem of computing an allocation of items that is approximately MMS based on the ordinal rankings of agents for items, and prove the results listed in Table~\ref{table:algorithm}.

\subsection{Identical Ordinary Preference and Allocation Sequence}

We first note that we can assume without loss of generality that all agents have {\em identical ordinary preference} (IDO), where a chore allocation instance is called IDO if $\sigma_i(k) = \sigma_j(k)$ for agents $i,j$ and index $k$. 
The original statement is proved for goods in \citep{journals/aamas/BouveretL16} and \citep{journals/teco/BarmanK20}, which is then adapted to  chores by \citep{corr/HuangL19}.

\begin{lemma}[\citep{corr/HuangL19}]
\label{lem:ido}
Suppose that there is an algorithm that runs in $T(n,m)$ time and returns an
$\alpha$-MMS allocation for all IDO instances. Then, there is an algorithm running in time $T(n,m)+O(nm\log m)$ outputing an $\alpha$-MMS  allocation for all instances that are not necessarily IDO.
\end{lemma}

We provide some high-level ideas for the proof of Lemma~\ref{lem:ido} as follows. For a formal proof, please refer to \citep{corr/HuangL19}.
For any instance $\cal I$ with parameters $N,M,c,\sigma$ that is not IDO, 
we create a corresponding IDO instance $\cal I'$  where 
the costs are defined as $c'_{ij} = c_{i, \sigma_i(j)}$ for all $i\in N$ and $j\in M$.
In other words, in $\cal I'$, item 1 is most costly and $m$ is least costly to every agent. 
Consequently, the resulting instance is IDO; moreover, the MMS values do not change.
Suppose we have an $\alpha$-approximation algorithm for IDO instances $\cal I'$.
Let $\pi_j \in N$ be the agent that receives item $j$ in the allocation. 
Then we have a length-$m$ sequence of ``picking ordering'' of agents $(\pi_m,\ldots,\pi_1)$.
Going back to $\cal I$, 
if we let agent $\pi_j$ pick her favorite unselected item (with lowest cost) in the order of $j = m, m-1, \ldots, 2, 1$, each agent's cost will not be higher than her cost in $\cal I'$ and thus the resulting allocation is also $\alpha$-MMS.

\smallskip

Accordingly, in the following, it suffices to only focus on IDO instances.
Assume items are ordered decreasingly regarding their costs: for any agent $i\in N$, we have 
\begin{equation*}
c_{i1}\geq c_{i2}\geq \ldots \geq c_{im}.
\end{equation*}
To simplify our statements, in this section we assume that $m\gg n$. Note that this is without loss of generality as we can append a sufficiently large number of items with cost $0$ for everyone to $M$.
The remaining part of this section focuses on the computation of  an allocation sequence $\pi \in N^m$ (a length-$m$ sequence of agents), where $\pi_j$ is the agent that receives item $j$. 
Since an allocation algorithm is uniquely defined by an allocation sequence, we use terms ``allocation algorithm'' and ``allocation sequence'' interchangeably.

\paragraph{Allocation sequence.}
One of the most well-known allocation sequences is {\em round-robin}, where the sequence is defined as $[1,\ldots,n,1,\ldots,n, \ldots]$.
That is, for $j = 1,2,\ldots, m$, we allocate item $j$ to agent $((j-1) \mod n) + 1$, until all items are allocated.
Observe that we can compactly represent the round-robin sequence as $\pi = [1,\ldots, n]^*$, which means that $\pi$ is obtained by repeating the \emph{pattern} $[1,\ldots,n]$ until the sequence has length $m$ (and the last replica may not be complete).
Like round-robin, in this paper we also focus on sequences with a certain pattern $p \in N^k$, for some $k\leq m$.
Formally speaking, the allocation sequence $\pi\in N^m$ with pattern $p\in N^k$ is obtained by repeating the pattern $p$ until $\pi$ has length $m$.
We denote the full sequence as $\pi = p^*$, and call it a {\em periodic} allocation sequence.

Recall that a round-robin algorithm achieves a $(2 - \frac{1}{n})$ approximation ratio \citep{aaai/AzizRSW17}.
In the following, we improve this approximation via a carefully designed periodic allocation sequence. 




\subsection{Upper Bounds}

In this section, we define the desired allocation sequences, and prove the approximation ratios (of MMS).
We first show the following technical lemma, which will be useful in the later analysis.

\begin{lemma}\label{lemma:two-out-of-every-k}
	Consider a sequence of items $S = \{j_1,j_2,\ldots,j_k\}$, ordered in descending order of costs.
	Suppose an agent $i$ receives two items $\{ j_x, j_k\}$ from $S$, where $x \geq \frac{k}{2}$.
	Then we have $c_{i,j_x} + c_{i,j_k} \leq \frac{2}{k}\cdot c_i(S)$.
\end{lemma}
\begin{proof}
	For convenience, let $a = c_{i,j_x}$ and $b=c_{i,j_k}$, where $a\geq b$.
	We have
	\begin{equation*}
	c_i(S) \geq x\cdot a + (k-x)\cdot b,
	\end{equation*}
	which implies
	\begin{equation*}
	\frac{c_{i,j_x} + c_{i,j_k}}{c_i(S)} \leq \frac{a+b}{x\cdot a + (k-x)\cdot b}  
	= \frac{a+b}{k\cdot b + x\cdot (a-b)}
	\leq \frac{a+b}{k\cdot b + \frac{k}{2}\cdot (a-b)} = \frac{2}{k},
	\end{equation*}
	where the second inequality follows from $x\geq \frac{k}{2}$. 
\end{proof}

Next, we define a periodic allocation algorithm, called {\em Sesqui-Round Robin} ($\SRR$), where the length of the repeating pattern is roughly $1.5n$.

\begin{algorithm}[htbp]
	\caption{\textsf{Sesqui-Round Robin} Algorithm.\label{alg:SRR}}
	\textbf{Input}: IDO instance with $c_{i1}\geq c_{i2}\geq \ldots \geq c_{im}$ for all $i\in N$.
	
	Initialize: $X_i = \emptyset$ for all $i \in N$.
	
	Set $p = \left[1,2,\ldots,n-1,n,n,n-1,\ldots,\lfloor\frac{n}{2}\rfloor+1\right]$.
	
	\For{$j = 1,2,\ldots,m$}
	{	
		$a = (j - 1 \mod |p|) + 1$ and $X_{p(a)} = X_{p(a)} \cup \{j\}$.
	}
	
	\textbf{Output}: Allocation $X=(X_1,\ldots,X_n)$.
\end{algorithm}

\paragraph{\textsf{Sesqui-Round Robin} ($\SRR$).}
We define the pattern of the periodic allocation sequence as
\begin{equation*}
p = \left[1,2,\ldots,n-1,n,n,n-1,\ldots,\lfloor\frac{n}{2}\rfloor+1\right].
\end{equation*}
For example, for $n = 2$ agents, the full sequence is $\pi = [1,2,2]^*$; for $n = 3$ the sequence is $\pi = [1,2,3,3,2]^*$.
Since the items are ordered in non-increasing order of their costs and incentive is not a concern, $\SRR$ is essentially a heavy cost first sequential allocation algorithm according to the repeating pattern $p$.
Intuitively, within each pattern, (1) each agent from 1 to $n$ is assigned an item and this part is the same with round-robin;
(2) then each agent in the second half of $[n]$ is assigned one more item but according to the reverse order because they have advantage in (1).
The pseudocode is provided in Algorithm~\ref{alg:SRR}.

\begin{theorem}[Approximation Ordinal Algorithms]\label{th:ordinal-main}
	Algorithm $\SRR$ is
	\begin{itemize}
		\item $4/3$-approximate MMS for $n=2$;
		\item $7/5$-approximate MMS for $n=3$;
		\item $5/3$-approximate MMS for any $n\geq 4$.
	\end{itemize}
\end{theorem}


We prove Theorem~\ref{th:ordinal-main} by proving the following three lemmas.


\begin{lemma}\label{lemma:4/3-approx}
	$\SRR$ is $4/3$-approximate MMS for $n=2$. 
\end{lemma}
\begin{proof}
	For $n=2$, $\SRR$ has repeating pattern $[1,2,2]$.
	That is, we assign to agent $1$ item set $X_1 = \{1,4,7,\ldots\} = \{3k+1 \mid k \in \mathbb{Z}^+\} \cap M$ and 
	assign to agent $2$ item set $X_2 = \{ 2,3,5,6,8,9\ldots \} = \{3k+2, 3k+3 \mid k \in \mathbb{Z}^+\} \cap M$.\footnote{$\mathbb{Z}^+$ represents the set of all non-negative integers $\{0,1,2,\ldots\}$.}
	
	Recall that items are indexed in descending order of costs.
	Let us first consider agent $1$ and define $f := c_{11}/\MMS_1$.
	By the second statement in Lemma~\ref{lem:mms:bound}, we have $\MMS_1 \ge c_{11}$ and thus $f\in[0,1]$.
	Note that after receiving item $1$, agent $1$ gets the last one out of every three consecutive items.
	Since $c_{1,3j-1} \ge c_{1,3j} \ge c_{1,3j+1}$ for all $j = 1, \ldots, \lfloor \frac{m-1}{3} \rfloor$, then
	\[
	3 \cdot \sum_{j=1}^{\lfloor \frac{m-1}{3} \rfloor} c_{1,1+3j} \le \sum_{j=1}^{\lfloor \frac{m-1}{3} \rfloor} c_{1,3j-1} + c_{1,3j} + c_{1,3j+1} = c_1(M)-c_{11}.
	\]
	Thus
	\begin{equation*}
	c_1(X_1) = c_{11}+ \sum_{j=1}^{\lfloor \frac{m-1}{3} \rfloor} c_{1,1+3j} \leq f\cdot \MMS_1 + \frac{1}{3}\cdot \left( c_1(M)-c_{11} \right).
	\end{equation*}
	By the first statement in Lemma~\ref{lem:mms:bound}, we have $c_1(M) \leq 2\cdot \MMS_1$ and thus
	\begin{equation*}
	c_1(X_1) \leq \left(f + \frac{1}{3}\cdot (2-f) \right)\cdot \MMS_1
	= \frac{2}{3} (1+f)\cdot \MMS_1 \leq \frac{4}{3}\cdot \MMS_1.
	\end{equation*}

	\smallskip
	
	Next we consider agent $2$.
	Similarly, since agent $2$ receives two items (of smallest cost) out of every three consecutive items, 
	and $c_{2,3j-2} \ge c_{2,3j-1} \ge c_{2,3j}$ for all $j = 1, \ldots, \lfloor \frac{m}{3} \rfloor$, we have
	\[
	c_2(X_2) \leq \frac{2}{3}\cdot c_2(M) \leq \frac{4}{3}\cdot \MMS_2,
	\]
	where the inequality also comes from $c_2(M) \leq 2\cdot \MMS_2$.
\end{proof}

Next we consider the case when $n=3$.

\begin{lemma}\label{lemma:7/5-approx}
    $\SRR$ is $7/5$-approximate MMS for $n=3$. 
\end{lemma}
\begin{proof}
For $n=3$, the allocation sequence has pattern $[1,2,3,3,2]$.
In the following, we consider the three agents separately and the reasoning is similar to that of Lemma \ref{lemma:4/3-approx}. 

\smallskip
\noindent
{\bf Agent $1$.}
Let $c_{11} = f\cdot \MMS_1$, where $f\in[0,1]$.
Note that after receiving the first item, agent $1$ receives one out of every $5$ consecutive items.
Hence
\begin{align*}
c_1(X_1) & \leq f\cdot \MMS_1 + \frac{1}{5}\cdot \left( c_1(M)-c_{11} \right) \\
& \leq \left( f+\frac{1}{5}\cdot (3-f) \right)\cdot\MMS_1 \leq \frac{7}{5}\cdot\MMS_1,
\end{align*}
where the second inequality holds due to $c_1(M)\leq 3\cdot \MMS_1$.

\smallskip
\noindent
{\bf Agent $2$.}
Let $c_{22} = f\cdot \MMS_2$, where $f\in [0,1]$.
Note that after receiving item $2$, for every $t=1,2,\ldots$, among the $5$ consecutive items
\begin{equation*}
S_t = \{ 3 + 5(t-1), 4+5(t-1), \ldots, 7+5(t-1) \},
\end{equation*}
agent $2$ receives the third item $5+5(t-1)$ and the last item $7+5(t-1)$.

By Lemma~\ref{lemma:two-out-of-every-k}, the total cost of items agent $2$ receives after item $2$ is at most $\frac{2}{5}\cdot \sum_{j=3}^m c_{2j}$.
Hence we have
\begin{align*}
c_2(X_2) & \leq f\cdot \MMS_2 + \frac{2}{5}\cdot \left( c_2(M) - c_{21} - c_{22} \right) \\
& \leq \left( f+\frac{2}{5}\cdot (3-2f) \right)\cdot\MMS_2 \leq \frac{7}{5}\cdot\MMS_2.
\end{align*}

\smallskip
\noindent
{\bf Agent $3$.}
Let $c_{33} + c_{34} = f\cdot \MMS_3$.
Note that among the first four items $\{ 1,2,3,4 \}$, at least two of them must appear in the same bundle of the MMS allocation of agent $3$.
Hence we have $\MMS_3 \geq c_{33} + c_{34}$, which implies $f\in [0,1]$.
Also note that $c_{31} + c_{32} + c_{33} + c_{34} \geq 2\cdot (c_{33} + c_{34}) = 2f\cdot \MMS_3$.

Observe that after receiving items $3$ and $4$, agent $3$ receives two items (of smallest cost) out of every $5$ consecutive items.
Hence we have
\begin{align*}
c_3(X_3) & \leq f\cdot \MMS_3 + \frac{2}{5}\cdot \left( c_3(M) - \sum_{j=1}^4 c_{3j} \right) \\
& \leq \left( f+\frac{2}{5}\cdot (3-2f) \right)\cdot\MMS_3 \leq \frac{7}{5}\cdot\MMS_3.
\end{align*}

Hence all agents receive a bundle of cost at most $\frac{7}{5}$ times her MMS value, and the lemma follows.
\end{proof}

Finally, we show that the approximation ratio of $\SRR$ is at most $\frac{5}{3}$, for any $n\geq 4$.

\begin{lemma}\label{lemma:5/3-approx}
    $\SRR$ is $5/3$-approximate MMS for $n \ge 4$. 
\end{lemma}
\begin{proof}
	Recall that the repeating pattern of the sequence is
	\begin{equation*}
	\left[1,2,\ldots,n-1,n,n,n-1,\ldots,\lfloor\frac{n}{2}\rfloor+1 \right].
	\end{equation*}
	
	For convenience we let $k = 2n-\lfloor \frac{n}{2} \rfloor$ be the length of the pattern. Note that we have $k = \frac{3n}{2}$ when $n$ is even; $k = \frac{3n+1}{2}$ when $n$ is odd.
	Fix any agent $i\in[n]$, we show that the set of items $X_i$ agent $i$ receives satisfies $c_i(X_i) \leq \frac{5}{3}\cdot \MMS_i$.
	
	\smallskip
	\noindent
	{\bf Case-1: $i \leq \lfloor\frac{n}{2}\rfloor$.}
	The algorithm assigns to agent $i$ the following items:
	\begin{equation*}
	X_i = \{ i, i + k, i + 2k, \ldots \}.
	\end{equation*}	
	Let $c_{ii}  = f\cdot \MMS_i$, where $f \in [0,1]$.
	Observe that after receiving item $i$, agent $i$ gets the item with minimum cost out of every $k$ items.
	Hence we have
	\begin{align*}
	c_i(X_i) & \leq  f\cdot \MMS_i + \frac{1}{k}\cdot \sum_{j = i+1}^m c_{ij} \leq f\cdot\MMS_i + \frac{2}{3n}\cdot \left( c_i(M) - \sum_{j=1}^i c_{ij} \right)\\
	& \leq f\cdot\MMS_i + \frac{2}{3n}\cdot \left( n\cdot \MMS_i - i\cdot f\cdot\MMS_i \right) \\
	& \leq \left( f + \frac{2}{3} \right)\cdot \MMS_i \leq \frac{5}{3}\cdot\MMS_i.
	\end{align*}
	
	\smallskip
	\noindent
	{\bf Case-2: $\lfloor\frac{n}{2}\rfloor+1 \leq i \leq n-\frac{k-2}{4}$.}
	Note that agent $i$ receives item $i$ first, then for every $t = 1,2,\ldots$, among the $k$ items
	\begin{equation*}
	S_t = \{ i+(t-1)k+1,i+(t-i)k+2,\ldots,i+t\cdot k \},
	\end{equation*}
	agent $i$ receives item $i+(t-1)k + 2(n-i)+1$ (the $(2(n-i)+1)$-th item in $S_t$) and item $i+t\cdot k$ (the last item in $S_t$).
	Observe that for $i\leq n-\frac{k-2}{4}$,
	\begin{equation*}
	2(n-i)+1 \geq \frac{k-2}{2}+1 = \frac{k}{2}.
	\end{equation*}	
	Hence by Lemma~\ref{lemma:two-out-of-every-k}, for every $t = 1,2,\ldots$ we have
	\begin{equation*}
	c_{i,i+(t-1)k + 2(n-i)+1} + c_{i,i+t\cdot k} \leq \frac{2}{k}\cdot c_i(S_t).
	\end{equation*}	
	As before, let $v_i = f\cdot \MMS_i$, where $f\in [0,1]$. We have
	\begin{align*}
	c_i(X_i) \leq & f\cdot \MMS_i + \frac{2}{k}\cdot \sum_{j = i+1}^m c_{ij} 
	= f\cdot\MMS_i + \frac{2}{k}\cdot \left( c_i(M) - \sum_{j=1}^i c_{ij} \right) \\
	\leq & f\cdot\MMS_i + \frac{2}{k}\cdot \left( n\cdot\MMS_i - i\cdot f \cdot \MMS_i\right) \\
	= & \left( \frac{2n}{k} + (1-\frac{2i}{k})\cdot f \right)\cdot \MMS_i	\leq \left( 1+ \frac{2(n-i)}{k} \right)\cdot \MMS_i	.
	\end{align*}	
	For $k = 2n-\lfloor\frac{n}{2} \rfloor$ and $i \geq \lfloor \frac{n}{2} \rfloor+1$, we have $\frac{n-i}{k} \leq \frac{0.5 n}{1.5 n} = \frac{1}{3}$, which implies
	\begin{equation*}
	c_i(X_i) \leq \left( 1+\frac{2}{3} \right)\cdot \MMS_i = \frac{5}{3}\cdot \MMS_i.
	\end{equation*}
	
	\smallskip
	\noindent
	{\bf Case-3: $i \geq n-\frac{k-2}{4}+1$.}
	Note that agent $i$ receives items
	\begin{equation*}
	X_i = \{ i, 2n-i+1, i+k, 2n-i+1+k, i+2k, 2n-i+1+2k,\ldots \}.
	\end{equation*}	
	In other words, agent $i$ receives items $i$ and $2n-i+1$ first, then for every $t = 1,2,\ldots$, among the $k$ items
	\begin{equation*}
	S_t = \{ 2n-i+2+(t-1)k, 2n-i+3+(t-i)k,\ldots,2n-i+1+t\cdot k \},
	\end{equation*}
	agent $i$ receives item $i+t\cdot k$ (the $(k-2(n-i)-1)$-th item in $S_t$) and item $2n-i+1+t\cdot k$ (the last item in $S_t$).	
	Observe that for $i\geq n-\frac{k-2}{4}+1$,
	\begin{equation*}
	k-2(n-i)-1 \geq k - 2(\frac{k-2}{4}-1)-1 = \frac{k}{2}+2 > \frac{k}{2}.
	\end{equation*}	
	Hence by Lemma~\ref{lemma:two-out-of-every-k}, for every $t = 1,2,\ldots$, the two items agent $i$ receives from $S_t$ have total cost	at most $\frac{2}{k}\cdot c_i(S_t)$.	
	Next, we bound the total cost $c_i(X_i)$ of agent $i$, taking into account the first two items agent $i$ receives.
	
	Let $v_i = f_1 \cdot \MMS_i$ and $v_{2n-i+1} = f_2 \cdot \MMS_i$, where $1\geq f_1 \geq f_2\geq 0$.
	
	\begin{claim}\label{claim:f1-and-f2}
		We have either $f_1 + f_2 \leq 1$ or $f_2 \leq \frac{1}{3}$.
	\end{claim}
	For continuity of presentation, we defer the proof of Claim \ref{claim:f1-and-f2} to the end of this subsection.
	By definition of $f_1$ and $f_2$ we have
	\begin{align*}
	c_i(X_i) \leq &\; f_1\cdot \MMS_i + f_2\cdot \MMS_i + \frac{2}{k}\cdot \sum_{j=2n-i+2}^m c_{ij} \\
	\leq &\; (f_1+f_2)\cdot \MMS_i + \frac{2}{k}\cdot \left( n\cdot\MMS_i -  \sum_{j=1}^{2n-i+1} c_{ij} \right)
	\end{align*}
	
	Note that for all $j\leq 2n-i+1$, we have $c_{ij} \geq f_2\cdot \MMS_i$; for all $j\leq i$, we have $c_{ij} \geq f_1\cdot \MMS_i$.
	Hence we have
	\begin{equation*}
	\sum_{j=1}^{2n-i+1} c_{ij} \geq i\cdot \Big(f_1 + (2n-2i+1)\cdot f_2 \Big)\cdot \MMS_i,
	\end{equation*}
	which implies
	\begin{align*}
	\frac{c_i(X_i)}{\MMS_i} \leq &\;  f_1 + f_2 + \frac{2}{k}\cdot \Big(n - i\cdot f_1 - (2n-2i+1)\cdot f_2 \Big) \\
	= &\; \frac{2n}{k} + \frac{k-2i}{k}\cdot f_1 + \frac{k - 2(2n-2i+1)}{k}\cdot f_2.
	\end{align*}
	
	Observe that the coefficient of $f_2$ is always positive since
	\begin{equation*}
	2n-2i+1 \leq 2n - 2(n-\frac{k-2}{4}+1)+1 = \frac{k}{2}-2 < \frac{k}{2}.
	\end{equation*}
	
	If $2i \geq k$, then the coefficient of $f_1$ is non-positive, and thus the maximum of RHS is achieved when $f_1 = f_2$.
	Note that when $f_1 = f_2$, by Claim~\ref{claim:f1-and-f2}, we have $f_2 \leq \frac{1}{2}$, which implies
	\begin{align*}
	\frac{c_i(X_i)}{\MMS_i} \leq &\; \frac{2n}{k} + \frac{2k - 4n + 2i - 1}{k}\cdot f_2 \\
	\leq &\; \frac{4n}{2k} + \frac{2k - 4n + 2i - 1}{2k} = 1+\frac{2i-1}{2k} < 1+\frac{2n}{1.5 n} = \frac{5}{3}.
	\end{align*}
	
	If $2i < k$, then using the fact that $i \geq n-\frac{k-2}{4}+1$, we have
	\begin{align*}
	\frac{c_i(X_i)}{\MMS_i} \leq &\; \frac{2n}{k} + \frac{k-2i}{k}\cdot f_1 + \frac{k - 2(2n-2i+1)}{k}\cdot f_2 \\
	\leq &\; \frac{2n}{k} + \frac{k-2n+\frac{k-2}{2}-2}{k}\cdot f_1 + \frac{k - 2(2n-k+1)}{k}\cdot f_2\\
	= &\; \frac{2n}{k} + \frac{3k-4n-6}{2k}\cdot f_1 + \frac{3k - 4n-2}{k}\cdot f_2 \\
	\leq &\; \frac{4}{3} + \frac{1}{6}\cdot f_1 + \frac{1}{3}\cdot f_2 = \frac{4}{3} + \frac{1}{3}\cdot (\frac{f_1}{2} +f_2).
	\end{align*}
	where the last inequality holds since $k\geq 1.5 n$.
	It not difficult to check that by Claim~\ref{claim:f1-and-f2}, $\frac{f_1}{2} + f_2 \leq 1$, which implies $\frac{c_i(X_i)}{\MMS_i} \leq \frac{4}{3} + \frac{1}{3} = \frac{5}{3}$.
\end{proof}

Combining Lemmas~\ref{lemma:4/3-approx},~\ref{lemma:7/5-approx} and~\ref{lemma:5/3-approx}, we have proved Theorem~\ref{th:ordinal-main}.
It remains to prove Claim \ref{claim:f1-and-f2}.

\begin{proofof}{Claim \ref{claim:f1-and-f2}}
	We call items $\{ 1,2,\ldots,i \}$ \emph{heavy} items and items $\{ i+1, i+2,\ldots,2n-i+1 \}$ \emph{light} items.
	Note that every heavy item must have cost at least $f_1\cdot \MMS_i$ and every light item must have cost at least $f_2\cdot \MMS_i$.
	Now consider the MMS allocation of agent $i$.
	If there exists a bundle containing both heavy and light items, or two heavy items, then we have
	\begin{equation*}
	\MMS_i \geq f_1\cdot \MMS_i + f_2\cdot \MMS_i,
	\end{equation*}
	which implies $f_1 + f_2 \leq 1$.
	Otherwise we know that if a bundle contains a heavy items, then it is a singleton.
	Note that there are $i$ heavy items, $2(n-i)+1$ light items and $n$ bundles.
	Hence there must exists a bundle containing three light items, which implies $\MMS_i \geq 3f_2\cdot \MMS_i$ and thus $f_2 \leq \frac{1}{3}$.
\end{proofof}

\subsection{Lower Bounds}

In the following, we give hard instances showing that the approximation ratios we obtained for $n \leq 3$ are optimal for deterministic ordinal algorithms.

\begin{theorem}[Hardness for Deterministic Algorithms]\label{th:hardness-ordinal}
	No deterministic ordinal algorithm has approximation ratio (w.r.t. MMS) smaller than
	\begin{itemize}
		\item $4/3$ for $n=2$;
		\item $7/5$ for $n=3$.
	\end{itemize}
\end{theorem}
\begin{proof}
	We first give a hard instance for $n=2$.
	Consider the instance in which the $2$ agents have identical ranking on $m=4$ items $\{1,2,3,4\}$.
	Without loss of generality, assume the first item (with maximum cost) is given to agent $1$.
	If the agent $1$ is allocated only one item, then for the case when $c_2 = (1,1,1,1)$, the approximation ratio is $\frac{3}{2}$ since the agent $2$ has total cost $3$ while $\MMS_{2}=2$.
	Otherwise (agent $1$ gets $\geq 2$ items), for the case when $c_1 = (3,1,1,1)$, the approximation ratio is at least $\frac{4}{3}$, as agent $1$ has total cost at least $3+1 = 4$ while $\MMS_{1}=3$.
	
	\smallskip
	
	Next, we consider the case when $n=3$.
	Suppose there exists an allocation that is strictly better than $7/5 = 1.4$-approximate.
	Let $1.4 - \epsilon$ be the approximation ratio of the algorithm, where $\epsilon \in (0, 0.4)$.
	In the following we consider a few instances with $m \geq \frac{2}{\epsilon}$ items, in which the $3$ agents have identical ranking on the items.
	For convenience of discussion we fix $m$ to be an odd number.
	
	First, observe that the first three items must be allocated to three different agents, otherwise the approximation is at least $1.5$.
	Without loss of generality, suppose item $i\in \{1,2,3\}$ is allocated to agent $i$.
	Then item $4$ must be allocated to agent $3$, as otherwise when all agents have cost function $(2,2,1,1,0,\ldots,0)$, the approximation ratio is $1.5$.
	Next, we consider how the items $M' = \{ 5,6,\ldots,m \}$ are allocated.
	Let $y_1, y_2$ and $y_3$ be the number of items in $M'$ allocated to item $1,2$ and $3$, respectively.
	
	\smallskip
	\noindent
	{\bf Agent-$1$.}
	Consider the instance in which the cost function of agent $1$ is
	\begin{equation*}
	c_1 = (1, \frac{2}{m-1}, \frac{2}{m-1}, \ldots, \frac{2}{m-1}).
	\end{equation*}
	Note that since $m$ is odd, we have $\MMS_1 = 1$.
	To ensure an approximation ratio of $1.4 - \epsilon$, we have
	$c_1(X_1) = 1 + \frac{2\cdot y_1}{m-1} \leq 1.4 - \epsilon$,
	which implies
	\begin{equation*}
	y_1 \leq \frac{m-1}{2}\cdot (0.4-\epsilon) < 0.2\cdot m - 0.5\cdot \epsilon.
	\end{equation*}
	
	\smallskip
	\noindent
	{\bf Agent-$2$.}
	Now consider the instance in which
	\begin{equation*}
	c_2 = (1, 1, \frac{1}{m-2}, \frac{1}{m-2}, \ldots, \frac{1}{m-2}).
	\end{equation*}
	Note that $\MMS_2 = 1$.
	To ensure an approximation ratio of $1.4 - \epsilon$, we have
	$c_2(X_2) = 1 + \frac{y_2}{m-2} \leq 1.4 - \epsilon$,
	which implies
	\begin{equation*}
	y_2 \leq (m-2)\cdot (0.4-\epsilon) < 0.4\cdot m - \epsilon\cdot m \leq 0.4\cdot m - 2,
	\end{equation*}
	where the last inequality follows from $m\geq \frac{2}{\epsilon}$.
	
	\smallskip
	\noindent
	{\bf Agent-$3$.}
	Finally, we consider the instance in which
	\begin{equation*}
	c_3 = (\frac{1}{2}, \frac{1}{2}, \frac{1}{2}, \frac{1}{2}, \frac{1}{m-3}, \frac{1}{m-3}, \ldots, \frac{1}{m-3}).
	\end{equation*}
	Since there are $m-4$ items with cost $\frac{1}{m-3}$, and $m$ is odd, it is not difficult to verify that $\MMS_3 = 1$.
	To ensure an approximation ratio of $1.4 - \epsilon$, we have
	$c_3(X_3) = 1 + \frac{y_3}{m-3} \leq 1.4 - \epsilon$,
	which implies
	\begin{equation*}
	y_3 \leq (m-3)\cdot (0.4-\epsilon) < 0.4\cdot m - \epsilon\cdot m \leq 0.4\cdot m - 2.
	\end{equation*}
	
	However, observe that now we have $y_1 + y_2 + y_3 < m - 4$,
	which is a contradiction since there are $m-4$ items in $M'$.
\end{proof}

Combining Theorem \ref{th:hardness-ordinal} and Lemmas \ref{lemma:4/3-approx} and \ref{lemma:7/5-approx},
we have shown that our algorithm is optimal for $n=2$ and $n=3$.
It would be natural to conjecture that the algorithm achieves optimal approximation ratios for larger $n$.
Unfortunately, this is not true. 
We defer this discussion to Section \ref{sec:conclusion}.


\section{Strategyproof Maximin Share Allocations}\label{sec:sp}


In this section, we take a mechanism design perspective and design strategyproof algorithms that can also approximate MMS fairness.
We first note that periodic sequential picking algorithm is unlikely to be strategyproof.
The following example shows that round-robin cannot guarantee strategyproofness, even on two agents.

\begin{example}
	Suppose there are two agents and four items.
	The first agent has ranking $c_{11} < c_{12} < c_{13} < c_{14}$ on the items, in the ascending order of costs.
	The second agent has ranking $c_{24} < c_{22} < c_{21} < c_{23}$.
	Suppose that both agents report truthfully then the algorithm allocates items $\{1,2\}$ to agent $1$ and items $\{3,4\}$ to agent $2$.
	However, if the second agent report differently as $c_{22} < c_{24} < c_{21} < c_{23}$, then the algorithm will allocate items $\{1,3\}$ to agent $1$ and items $\{2,4\}$ to agent $2$.
	In other words, agent $2$ receives strictly better allocation by misreporting, and hence the algorithm is not strategyproof.
\end{example}

\subsection{Deterministic Algorithm}\label{ssec:det}

We present a deterministic sequential picking algorithm that is $O(\log\frac{m}{n})$-approximate and strategyproof.
Recall that when items are goods, \citep{ijcai/AmanatidisBM16} gave a deterministic $O(m-n)$-approximate strategyproof ordinal algorithm. 
In the following, we show that if all the items are chores, the approximation ratio can be to $O(\log \frac{m}{n})$.
Without loss of generality, we assume that $n$ and $m/n$ are at least some sufficiently large constant.
As otherwise it is trivial to obtain an $O(1)$-approximation by assigning $m/n$ arbitrary items to each agent.

\begin{theorem}\label{th:strategyproof-ordinal}
	There exists a deterministic strategyproof ordinal algorithm with approximation ratio $O(\log \frac{m}{n})$.
\end{theorem}

We first define another typical sequential picking algorithm, where each agent has a single chance to select items. 

\paragraph{$\CP$.}
Fix a sequence of integers $a_1, \ldots, a_n$ such that $\sum_{i\leq n}a_i = m$.
Order the agents arbitrarily. For $i=n,n-1,\ldots,1$, let agent $i$ pick $a_i$ items from the remaining items.
We do not restrict which items each agent should pick, but of course strategic agents want to select items with smallest cost.
The pseudocode is provided in Algorithm~\ref{alg:CP}.
Recall that $\sigma_i(1)$ is the least preferred item of agent $i$, and $\sigma_i(m)$ is the most preferred.

\begin{algorithm}[htbp]
\caption{$\CP$ Algorithm.\label{alg:CP}}

\textbf{Parameters}: Integers $a_1, \ldots, a_n$ such that $\sum_{i\leq n}a_i = m$.

\textbf{Input}: The ordinal preference $\sigma$ of agents.

Initialize: $X_i = \emptyset$ for all $i \in N$.

\For{$i = 1,2,\ldots,n$}
{	
	\For{$j = 1,2,\ldots,a_i$}
{	
	Let $e^* = \arg\max_{e\in M} \{\sigma^{-1}_{i}(e)\}$; Set $X_i = X_i \cup \{e^*\}$ and $M = M\setminus \{e^*\}$.
}	
}

\textbf{Output}: Allocation $X=(X_1,\ldots,X_n)$.
\end{algorithm}

We note that as long as $a_i$'s do not depend on the reported preferences of agents, the rule discussed above is the serial dictatorship rule for multi-unit demands.
When it is agent $i$'s turn to pick items, it is easy to see that her optimal strategy is to pick the top-$a_i$ items with the smallest cost, among the remaining items.
Hence immediately we have the following lemma.

\begin{lemma}
	For any $\{a_i\}_{i\leq n}$, $\CP$ is strategyproof.
\end{lemma}

It remains to prove the approximation ratio.

\begin{lemma}\label{lem:ai}
	There exists a sequence $\{a_i\}_{i\leq n}$ such that the approximation ratio of $\CP$ is $O(\log \frac{m}{n})$.
\end{lemma}
\begin{proof}
	We first establish a lower bound on the approximation ratio in terms of $\{a_i\}_{i\leq n}$. Then we show how to fix the numbers appropriately to get a small ratio.
	Let $r$ be the approximation ratio of the algorithm.
	
	Consider the moment when agent $i$ needs to pick $a_i$ items.
	Recall that at this moment, there are $\sum_{j\leq i} a_j$ items, and the $a_i$ ones with the smallest cost will be chosen by agent $i$.
	Let $c$ be the average cost of items agent $i$ picks, i.e., $c_i(X_i) = c\cdot a_i$.
	On the other hand, each of the $\sum_{j\leq i-1} a_j$ items left has cost at least $c$.
	Thus we have $\MMS_i \geq c \cdot \left\lceil \frac{a_1+\ldots+a_{i-1}}{n} \right\rceil$ and
	\begin{equation*}
	r = \max_{i\in N}\left \{\frac{c_i(X_i)}{\MMS_i} \right\} \leq \max_{i\in N} \left\{\frac{a_i}{\left\lceil \frac{a_1+\ldots+a_{i-1}}{n} \right\rceil} \right\}.
	\end{equation*}
	
	It suffices to compute a sequence of $a_1,\ldots, a_n$ that sum to $m$ and minimize this ratio.
	Fix $K= 2\log \frac{m}{n}$. Let
	\begin{equation*}
	a_i = \begin{cases}
	2, & i\leq \frac{n}{2}, \\
	\min\{  m-\sum_{j<i}a_j, \left\lceil K\cdot (1+\frac{K}{n})^{i-\frac{n}{2}-1} \right\rceil \}, & i> \frac{n}{2}.
	\end{cases}
	\end{equation*}
	
	Note that the first term of $\min\{\cdot,\cdot\}$ is to guarantee we leave enough items for the remaining agents.
	Moreover, truncating $a_i$ is only helpful for minimizing the approximation ratio and thus we only need to consider the case when $a_i$ equals the second term of $\min\{\cdot,\cdot\}$.
	In the following, we show that
	\begin{enumerate}
		\item all items are picked: $\sum_{i\in N}a_i = m$;
		\item for every $i>\frac{n}{2}$: $a_i \leq K\cdot \left\lceil \frac{a_1+\ldots+a_{i-1}}{n} \right\rceil$.
	\end{enumerate}
	
	Note that for $i\leq \frac{n}{2}$, since agent $i$ receives $2$ items, the approximation ratio is trivially guaranteed.
	The first statement holds because
	\begin{align*}
	 & \sum_{i=1}^\frac{n}{2} 2 + \sum_{i=\frac{n}{2}+1}^n \left( K\cdot (1+\frac{K}{n})^{i-\frac{n}{2}-1} \right)
	 \\
	 =  &\sum_{i \leq \frac{n}{2}} \left( K\cdot (1+\frac{K}{n})^{i-1} \right) + n =   (1+\frac{K}{n})^\frac{n}{2}\cdot n - n + n \geq 2^\frac{K}{2}\cdot n^2 > m,
	\end{align*}
	and $a_i$'s will be truncated when their sum exceeds $m$.
	
	For $i>\frac{n}{2}$, observe that (let $l = i-\frac{n}{2}-1$)
	\begin{equation*}
	\frac{1}{n} \sum_{j=1}^{i-1} a_j  = 1 + \frac{1}{n}\sum_{j = 1}^{l} K\cdot (1+\frac{K}{n})^{j-1} 
	  = 1 + (1+\frac{K}{n})^{l} - 1 = (1+\frac{K}{n})^{l}.
	\end{equation*}	
	Thus we have
	\begin{equation*}
	\textstyle a_i \leq \left\lceil K\cdot (1+\frac{K}{n})^{l} \right\rceil \leq K\cdot \left\lceil  (1+\frac{K}{n})^{l} \right\rceil \leq K\cdot \left\lceil \frac{a_1+\ldots+a_{i-1}}{n} \right\rceil,
	\end{equation*}
	as claimed.
\end{proof}

We conclude this section by showing that our approximation ratio is asymptotically optimal for all $\CP$ algorithms.

\begin{lemma}[Limits of $\CP$]
	The $\CP$ algorithm (with any $\{a_i\}_{i\in N}$) has approximation ratio $\Omega(\log \frac{m}{n})$.
\end{lemma}
\begin{proof}
	Fix $K = \frac{1}{4}\log \frac{m}{n}$. Suppose there exists a sequence of $\{a_i\}_{i\in N}$ such that the algorithm is $K$-approximate.
	Then the last agent to act must receive at most $K$ items, i.e., $a_1 \leq K$.
	Next, we show by induction on $i=2,3,\ldots,n$ that $a_i \leq K(1+\frac{2K}{n})^{i-1}$ for all $i\in N$.
	Suppose the statement is true for $a_1,\ldots, a_i$. Then if $a_{i+1} > K(1+\frac{2K}{n})^{i}$, we have
	\begin{equation*}
	\frac{a_{i+1}}{a_1+\ldots+a_{i+1}} > \frac{ K(1+\frac{2K}{n})^{i}}{k\cdot\frac{n}{2K}((1+\frac{2K}{n})^{i+1}-1)} \geq \frac{K}{n}.
	\end{equation*}	
	Thus we have
	\begin{equation*}
	\textstyle \sum_{i=1}^n a_i \leq n\cdot\left( (1+\frac{2K}{n})^n - 1 \right)
	\leq n\cdot \left( e^{2K} - 1 \right) < m,
	\end{equation*}
	which is a contradiction, since not all items are allocated.
\end{proof}

\subsection{Randomized Algorithm} \label{ssec:ran}

Via a carefully designed $\CP$ algorithm, we obtained a logarithmic approximation for the problem.
However, the algorithm may still have poor performance when the number of items is much larger than the number of agents, e.g., $m = 2^n$.
In this section, we present a randomized $O(\sqrt{\log n})$-approximation ordinal algorithm,
which is strategyproof in expectation.

Basically, if we randomly allocate all the items, one is able to show that the algorithm achieves an approximation of $O(\log n)$.
The drawback of this na{\" i}ve randomized algorithm is that it totally ignores the rankings of agents.
In the following, we show that if the agents have opportunities to decline some ``bad'' items,
the performance of this randomized algorithm improves to $O(\sqrt{\log n})$.
Note that since we already have an $O(\log \frac{m}{n})$-approximate deterministic algorithm for the ordinal model,
it suffices to consider the case when $m \geq n\log n$.

\paragraph{$\RD$.}
Let $K=\lfloor n\sqrt{\log n} \rfloor$.
Based on the ordering of items submitted by agents, for each agent $i$,
we label the $K$ items with the largest cost as ``large'', and the remaining as ``small''.
It can also be regarded as each agent reports a set $M_{i}$ of large items with $|M_{i}|=K$.
The algorithm operates in two phases.
\begin{itemize}
	\item Phase 1: every item is allocated to a uniformly-at-random chosen agent, independently.
	After all allocations, gather all the large items assigned to every agent into set $M_b$. Note that $M_{b}$ is also a random set.
	
	\item Phase 2: Redistribute the items in $M_{b}$ evenly to all agents: every agent gets $|M_b|/n$ random items.
\end{itemize}
The pseudocode is provided in Algorithm~\ref{alg:RD}.

\begin{algorithm}[htbp]
\caption{$\RD$ Algorithm.\label{alg:RD}}

\textbf{Input}: The ordinal preference $\sigma$ of agents.

Initialize: $X_i = \emptyset$ for all $i \in N$ and $M_b = \emptyset$.

For each $i\in N$: let $M_{i} = \{ \sigma_i(1),\sigma_i(2),\ldots, \sigma_i(K) \}$, where $K = \lfloor n\sqrt{\log n} \rfloor$.


\For{$j = 1,2,\ldots,m$}
{	
	Randomly and uniformly select an agent $i$ and set $X_i = X_i \cup \{j\}$.
}

\For{$i = 1,2,\ldots,n$}
{	
	Set $M_b = M_b \cup (M_i \cap X_i)$ and $X_i = X_i \setminus M_i$.
}


Randomly divide $M_b$ into $n$ bundles $(Y_1, \ldots, Y_n)$, each with size $|M_b|/n$.

\For{$i = 1,2,\ldots,n$}
{	
	Set $X_i = X_i \cup Y_i$.
}

\textbf{Output}: Allocation $X=(X_1,\ldots,X_n)$.
\end{algorithm}

\begin{theorem}\label{th:strategyproof-ordinal:random}
	There exists a randomized strategyproof ordinal algorithm with approximation ratio $O(\sqrt{\log n})$.
\end{theorem}

We prove Theorem~\ref{th:strategyproof-ordinal:random} by proving the following Lemmas \ref{lem:sp:approx-ratio} and \ref{lem:sp:sp}.

\begin{lemma}\label{lem:sp:approx-ratio}
	In expectation, the approximation ratio of Algorithm $\RD$ is $O(\sqrt{\log n})$.
\end{lemma}
\begin{proof}
	We show that with probability at least $1-\frac{2}{n}$, every agent $i$ receives a collection of items of cost at most $O(\sqrt{\log n})\cdot \MMS_i$.
	Fix any agent $i$. 
	Without loss of generality, we order the items according to agent $i$'s ranking,
	i.e., $\sigma_{i}(j)=j$ for any $j\in M$ and $c_{i1}\geq \ldots \geq c_{im}$.
	
	For ease of analysis, we rescale the costs such that
	\begin{equation*}
	c_{i1}+c_{i2}+\ldots+c_{im} = n\sqrt{\log n} = K.
	\end{equation*}
	
	Note that after the scaling, agent $i$'s maximin share is $\MMS_{i}\geq \sqrt{\log n}$.
	Let $x_{ij}$ denote the random variable indicating the contribution of item $j$ to the cost of agent $i$.
	Then for $j > K$, $x_{ij} = c_{ij}$ with probability $\frac{1}{n}$, and $x_{ij} = 0$ otherwise.
	For $j \leq K$, $x_{ij} = 0$ with probability $1$.
	Note that
	\begin{equation*}
	\mathbf{E}[\sum_{i=1}^m x_i ] = \frac{1}{n}\cdot \sum_{i=K+1}^m c_{ij} \leq \frac{K}{n} = \sqrt{\log n}.
	\end{equation*}
	
	Moreover, we have $c_{ij} \leq 1$ for $j > K$, as otherwise we have the contradiction that $\sum_{j=1}^{K}c_{ij} > K$.
	Note that $\{x_{ij}\}_{j\leq m}$ are independent random variables taking value in $[0,1]$.
	Hence by Chernoff bound we have
	\begin{align*}
	& \Pr[ \sum_{j=1}^m x_{ij} \geq 7\sqrt{\log n} \cdot \MMS_{i} ] 
	\leq   \Pr[ \sum_{j=1}^m x_{ij} \geq 7\log n ] \\
	\leq &  \exp\left(-\frac{1}{3}\cdot \left(\frac{7\log n}{\mathbf{E}[\sum_{i=1}^m x_i]}-1\right) \cdot \mathbf{E}[\sum_{i=1}^m x_i] \right) < \frac{1}{n^2}.
	\end{align*}
	
	Then by union bound over the $n$ agents, we conclude that with probability at least $1-\frac{1}{n}$, every agent $i$ receives a bundle of items of cost at most $O(\sqrt{\log n})\cdot \MMS_i$ in Phase 1.
	
	Now we consider the items received by an agent in the second phase.
	Recall that the items $M_b$ will be reallocated evenly.
	By the second argument of Lemma \ref{lem:mms:bound}, to show that every agent $i$ receives a bundle of items of cost $O(\sqrt{\log n})\cdot \MMS_i$ in the second phase, it suffices to prove that $|M_b| = O(n\sqrt{\log n})$ (with probability at least $1-\frac{1}{n}$).
	
	Let $y_j\in\{0,1\}$ be the random variable indicating whether item $j$ is contained in $M_b$.
	For every item $j$, let $b_j=|\{k : j \in M_{k}\}|$ be the number of agents that label item $j$ as ``large''.
	Then we have $y_j = 1$ with probability $\frac{b_j}{n}$.
	Since every agent labels exactly $n\sqrt{\log n}$ items, we have
	\begin{equation*}
	\mathbf{E}[|M_b|] = \mathbf{E}[\sum_{i=1}^m y_i] = \frac{1}{n}\sum_{i=1}^m b_i = n\sqrt{\log n}.
	\end{equation*}
	Applying Chernoff bound we have
	\begin{align*}
	\Pr[ \sum_{i=1}^m y_i \geq 2n\sqrt{\log n}] \leq \exp\left(-\frac{n\sqrt{\log n}}{3} \right) < \frac{1}{n}.
	\end{align*}
	
	Thus, with probability at least $1-\frac{2}{n}$, every agent $i$ receives a bundle of items with cost $O(\sqrt{\log n}\cdot \MMS_i)$ in the two phases combined.
	Since in the worse case, $i$ receives a total cost of at most $n\cdot \MMS_i$, in expectation, the approximation ratio is $(1-\frac{2}{n})\cdot O(\sqrt{\log n}) + \frac{2}{n}\cdot n = O(\sqrt{\log n})$.
\end{proof}

\begin{lemma}\label{lem:sp:sp}
	$\RD$ is strategyproof in expectation.
\end{lemma}
\begin{proof}
	To prove that the algorithm is strategyproof in expectation, it suffices to show that for every agent, the expected cost she is assigned is minimized when being truthful.
	Let $K=n\sqrt{\log n}$ and fix any agent $i$.
	Suppose $c_{i1},\ldots,c_{iK}$ are the costs of items labelled ``large'' by the agent; and $c_{i,K+1},\ldots,c_{im}$ are the remaining items.
	Then the expected cost assigned to the agent in the first phase is given by $\frac{1}{n}\sum_{j=K+1}^m c_{ij}$, as every item is assigned to her with probability $\frac{1}{n}$.
	Next we consider the cost incurred to agents in the second phase.
	
	Recall that the expected total cost of items to be reallocated in the second phase is $\mathbf{E}[\sum_{j\in M_b} c_{ij}] = \sum_{j=1}^m c_{ij}\cdot \frac{b_j}{n}$, where $b_j$ is the number of agents that label item $j$ ``large''.
	Let $\mathcal{E}$ be this expectation when agent $i$ does not label any item ``large''.
	By labelling $c_{i1},\ldots, c_{iK}$ ``large'', agent $i$ increases the probability of each item $j\leq K$ being included in $M_b$ by $\frac{1}{n}$.
	Thus it contributes an $\frac{1}{n}\sum_{j=1}^K c_{ij}$ increase to the expectation of total cost of $M_b$.
	In other words,
	\begin{equation*}
	\mathbf{E}[\sum_{j\in M_b} c_{ij}] = \mathcal{E} + \frac{1}{n}\sum_{j=1}^K c_{ij}.
	\end{equation*}
	
	Since a random subset of $\frac{|M_b|}{n}$ items from $M_b$ will be assigned to agent~$i$, the expected total cost of items assigned to her in the two phases is given by
	\begin{equation*}
	\frac{1}{n}\sum_{j=K+1}^m c_{ij} + \frac{1}{n}\cdot \left( \mathcal{E} + \frac{1}{n}\sum_{j=1}^K c_{ij} \right).
	\end{equation*}
	
	Obviously, the expression is minimized when $c_{i1}+\ldots+c_{iK}$ is maximized.
	Hence every agent minimizes her expected cost by telling the true ranking.
\end{proof}

\section{Related Works}\label{sec:review}

The study of computing fair allocations of resources has a long history.
Arguably, two of the most widely studied solution concepts are envy-freeness (EF) and proportionality, whose existence is guaranteed when there is a single divisible item, i.e., the {\em cake cutting problem} \citep{BrTa96a,journals/combinatorics/Stromquist08,focs/AzizM16}.  
The problem becomes tricky when the items are indivisible, because exact envy-free or proportional allocations barely exist and are hard to approximate. 
In order to characterize the extent to which fairness can be guaranteed in the indivisible setting, several relaxations have been proposed, such as envy-free up to one item (EF1) \citep{conf/sigecom/LiptonMMS04}, envy-free up to any item (EFx) \citep{journals/teco/CaragiannisKMPS19}, and maximin share fair (MMS) \citep{bqgt/Budish10}, whose relations have been discussed by \citep{conf/ijcai/AmanatidisBM18}. 
Among these relaxations, MMS is undoubtedly one of the most widely studied one. 
It has been conjectured that an MMS allocation always exists until \citep{sigecom/ProcacciaW14,jacm/KurokawaPW18} identified a counter-example.
Thereafter, there appeared rich works designing approximate MMS allocations. 
The first $O(1)$ approximation algorithm was given by \citep{sigecom/ProcacciaW14}, whose approximation ratio is 3/2 but its running time can be exponential in the number of agents.
Later, \citep{icalp/AmanatidisMNS15,journals/talg/AmanatidisMNS17} refined the algorithm in \citep{sigecom/ProcacciaW14} and guaranteed the same approximation with a polynomial running time.
The same approximation is also obtained in \citep{conf/soda/GargMT19,conf/sigecom/BarmanM17,journals/teco/BarmanK20}.
\citep{ec/GhodsiHSSY18} improved these results by giving a 4/3 approximation algorithm whose running time may be exponential. 
More recently, \citep{sigecom/GargT20} designed a polynomial time algorithm to find a 4/3 approximate MMS allocation and proved the existence of $(4/3- \Theta(1/n))$-MMS allocation, breaking the barrier of 4/3.

Although most of the work on MMS allocation of items is for the case of goods, recently, fair allocation of chores~\citep{aaai/AzizRSW17} or combinations of goods and chores~\citep{ijcai/AzizCIW19} have received much attention.
\citep{aaai/AzizRSW17} proved that MMS allocations do not always exist but can be easily 2-approximated. 
Later, \citep{ec/BarmanM17} presented a $4/3$-approximation algorithm for MMS allocation of chores, and \citep{corr/HuangL19} further improved this ratio to $11/9$.
\citep{ijcai/AzizCL19} extended the definition of MMS to the weighted version that deals with asymmetric agents.

\paragraph{Distortion.}
Our work is also inspired by the growing literature on distortion in voting, 
where voters express ordinal preferences (instead of numerical utilities) over candidates \citep{conf/cia/ProcacciaR06,journals/ai/BoutilierCHLPS15,journals/jair/CaragiannisNPS17,conf/sigecom/MandalSW20}; and matching, where only the edge ranking is known instead of the exact weights \citep{DBLP:conf/wine/AnshelevichS16,DBLP:journals/sigecom/Anshelevich16,DBLP:conf/aaai/AnshelevichS16}.
The goal is to use the partial information to find solutions that  maximizes social welfare, and {\em distortion} is the measure to evaluate the worst-case multiplicative loss in social welfare due to this lack of information.
A major focus of our work is identifying what approximation guarantees of fairness can be achieved by only using ordinal information, which is naturally connected to the work on distortion.
There has been a substantial amount of work on using ordinal preferences in fair allocation of indivisible goods.
For example, \citep{ai/AzizGMW15} considered the question of checking the existence of allocations that possibly or necessarily satisfy certain fairness guarantees such as envy-freeness given only ordinal preferences of the agents over the goods.
\citep{ecai/BouveretEL10} studied similar questions, but given partial ordinal preferences of the agents over bundles of goods.
More closely related to ours are the papers that use ordinal allocation rules (such as picking sequence rules) in
settings with cardinal valuations. For example, \citep{conf/ecai/AzizKWX16} focused on the complexity of checking what social welfare such rules can possibly or necessarily achieve. \citep{ijcai/AmanatidisBM16} sought to use picking sequence rules to obtain approximation of the MMS fairness. 
Very recently, \citep{halperndistortion} showed that there is an algorithm using ordinal preferences to guarantee $O(\log n)$-approximate MMS fairness when items are goods.

\paragraph{Mechanism Design without Money.}
Strategyproofness is a challenging property to satisfy for fair division algorithms. For cake cutting problem, \citep{journals/geb/ChenLPP13} and \citep{conf/ijcai/BeiCHTW17} studied to what extent there exist strategyproof algorithms to fairly allocate the cake for piece-wise uniform or linear valuations.  
\citep{conf/wine/MayaN12} provided a characterization of strategyproof algorithms for the case of two agents.
When items are indivisible, \citep{conf/aldt/CaragiannisKKK09} and \citep{conf/sigecom/LiptonMMS04} have discussed how to elicit true information from the agents while ensuring some degree of envy-freeness.
More recently, \citep{ijcai/AmanatidisBM16} initiated the work on strategyproof allocation of goods with respect to MMS fairness. 
One important algorithm class is sequential picking, which is a generalization of round-robin.
\citep{KoCh71a,conf/ecai/BouveretL14,conf/aaai/AzizBLM17,conf/aldt/0001GW17} studied strategic aspects of sequential picking.
There is also work on the approximation of welfare that can be achieved by strategyproof algorithms for allocation of \emph{divisible} items~(e.g., \citep{atal/AzizFCMM16,DBLP:conf/sigecom/ColeGG13}).
All the above research focuses on the case of goods. In this paper, we aim at addressing the incentive compatibility for allocating indivisible chores.

\section{Discussion and Conclusion}
\label{sec:conclusion}

\paragraph{$\SRR$ Is Not Optimal for Larger $n$.}
As we have proved in Section \ref{sec:algorithmic}, our algorithm $\SRR$ achieves optimal approximation ratios for $n=2$ and $n=3$.
However, it fails to return an optimal solution when $n = 4$.
Actually, following similar analysis for $n=2$ and $n=3$, one can show that the approximation ratio of our algorithm is $1.5$ for $n=4$.
However, we are aware of an algorithm that performs strictly better than $1.499$-approximate.
Furthermore, we are aware of an instance with $n=4$, for which no ordinal algorithm performs better than $1.405$-approximate.
To this end, we conjecture that the optimal approximation ratio $r^*(n)$ (with $n$ agents) is an increasing function of $n$.
In this paper we have shown that
\begin{equation*}
r^*(2) = \frac{4}{3}\approx 1.333,\quad r^*(3) = \frac{7}{5}=1.4,\text{ and}  \quad \forall n,\ r^*(n)\leq \frac{5}{3}\approx 1.667 .
\end{equation*} 
We can also show that $
1.405 < r^*(4) < 1.499$ \footnote{Since we are not able to obtain the exact ratio, we did not include the analysis here.}.
We leave it as a future work to analyze the optimal ratio $r^*(n)$ for $n\geq 4$.

\paragraph{Constant Approximations for Our Strategyproof Algorithm.}
We have shown in Section~\ref{ssec:det} a deterministic strategyproof algorithm that is $O(\log(m/n))$-approximate MMS.
However, in many applications it is desirable to obtain constant approximation ratios.
While our algorithm has constant approximation ratios when $m = O(n)$, it is not clear how large the constant is.
In particular, if we need to guarantee an approximation ratio $r$, what is the maximum number of items we can handle?
In this part we give a detailed analysis to answer this question.
Following the analysis of Section~\ref{ssec:det}, in order to guarantee an approximation ratio of $r$, we can set $a_1 = r$, and for each $i = 2,\ldots,n$, we set $a_i = r\cdot \left\lceil \frac{a_1 + \ldots + a_{i-1}}{n} \right\rceil$.
To guarantee that all items are allocated, we have $m\leq \sum_{i=1}^n a_i$.
For example, if $r = 2$, we have
\begin{align*}
	a_1  =\ldots = a_{\frac{n}{2}} = 2, \qquad
	& a_{\frac{n}{2}+1} = \ldots = a_{\frac{3n}{4}} = 4, \\
	a_{\frac{3n}{4}+1} = \ldots = a_{\frac{11n}{12}} = 6, \qquad
	& a_{\frac{11n}{12}+1} = \ldots = a_n = 8.
\end{align*}
Hence we have $m \leq \sum_{i=1}^n a_i = \frac{11}{3}n\approx 3.67 n$.
Similarly, to guarantee an approximation of $r = 3$, we can let the first $\frac{n}{3}$ values of $a_i$ be $3$; the next $\frac{n}{6}$ values of $a_i$ be $6$; then the next $\frac{n}{9}$ values of $a_i$ be $9$, etc.
Following similar calculations, one can verify that the maximum number of items the algorithm can handle to guarantee $r = 3$ is $m \approx 10.26 n$; for $r = 4$, we have $m \approx 30.15 n$.

\paragraph{Conclusion} 
In this paper, we initiated the study of approximate and strategyproof maximin fair algorithms for chore allocation using ordinal preferences.
Our study leads to several new questions.
Two most obvious research questions are to find the optimal ordinal algorithm for arbitrary number of agents,
and to improve the approximation or study the lower bounds of strategyproof (randomized) algorithms.
At present, we have two parallel lines of research for goods and chores.
It is important to consider similar questions for combinations of goods and chores~\citep{ijcai/AzizCIW19}.
Finally, it is interesting to extend our work to the case of asymmetric agents~\citep{ijcai/AzizCL19}, 
where agents possess different weights and a fair allocation should respect these weights. 

%

\bibliographystyle{plainnat}
\bibliography{sp-chores}

\end{document}